\documentclass[submission,copyright,creativecommons]{eptcs}
\providecommand{\event}{TARK 2023} 
\usepackage{amsmath, amsthm, amssymb, mathrsfs, pifont}
 \usepackage{tikz}
 \usepackage{dsfont}
\usepackage{thmtools, thm-restate}
\usepackage{tabularx, makecell, array, diagbox}
\usepackage[title,titletoc]{appendix}
\usepackage{enumerate}
\newtheorem{theorem}{Theorem}[section]
\newtheorem{lemma}[theorem]{Lemma}
\newtheorem{proposition}[theorem]{Proposition}

\theoremstyle{definition}
\newtheorem{definition}[theorem]{Definition}

\theoremstyle{remark}

\newcommand{\sembracki}[1]{\left[\!\left[#1\right]\!\right]}
\newcommand{\xmark}{\ding{55}}

\usepackage{iftex}

\ifpdf
  \usepackage{underscore}         
  \usepackage[T1]{fontenc}        
\else
  \usepackage{breakurl}           
\fi

\def\titlerunning{A Logic-Based Analysis of Responsibility}
\def\authorrunning{A.I.~Ram\'{i}rez~Abarca}

\title{A Logic-Based Analysis of Responsibility}
\author{Aldo Iv\'an Ram\'irez Abarca
\institute{Utrecht University\\ Utrecht, The Netherlands}
\email{nadabundo@gmail.com}}
\begin{document}
\maketitle

\begin{abstract}
This paper presents a logic-based framework to analyze responsibility, which I refer to as intentional epistemic act-utilitarian stit theory (IEAUST). To be precise, IEAUST is used to model and syntactically characterize various modes of responsibility, where by `modes of responsibility' I mean instances of Broersen's three categories of responsibility (causal, informational, and motivational responsibility), cast against the background of particular deontic contexts. IEAUST is obtained by integrating a modal language to express the following components of responsibility on stit models: agency, epistemic notions, intentionality, and different senses of obligation. With such a language, I characterize the components of responsibility using particular formulas. Then, adopting a compositional approach---where complex modalities are built out of more basic ones---these characterizations of the components are used to formalize the aforementioned modes of responsibility. 
\end{abstract}

\section{Introduction}

The study of responsibility is a complicated matter. The term is used in different ways in different fields, and it is easy to engage in everyday discussions as to why someone should be considered responsible for something. Typically, the backdrop of these discussions involves social, legal, moral, or philosophical problems, each with slightly different meanings for expressions like \emph{being responsible for...}, \emph{being held responsible for...}, or \emph{having the responsibility of...}, among others. Therefore---to approach such problems efficiently---there is a demand for clear,  taxonomical definitions of responsibility.   

For instance, suppose that you are a judge in Texas. You are presiding over a trial where the defendant is being charged with first-degree murder. The alleged crime is horrible, and the prosecution seeks capital punishment. The case is as follows: driving her car, the defendant ran over a traffic officer that was holding a stop-sign at a crossing walk, while school children were crossing the street. The traffic officer was killed, and some of the children were severely injured. A highly complicated case, the possibility of a death-penalty sentence means that the life of the defendant is at stake. More than ever, due process is imperative. As the presiding judge, you must abide by the prevailing definitions of criminal liability with precision. In other words, there is little to no room for ambiguity in the ruling, and your handling of the notions associated with responsibility in criminal law should be impeccable.  

As this example suggests, a framework with intelligible, realistically applicable definitions of responsibility is paramount in the field of law. However, responsibility-related problems arise across many other disciplines---social psychology, philosophy of emotion, legal theory, and ethics, to name a few \cite{lorini2014logical,weiner1995judgments}. A clear pattern in all these is the intent of issuing standards for when---and to what extent---an agent should be held responsible for a state of affairs. 

This is where Logic lends a hand. The development of expressive logics---to reason about agents' decisions in situations with moral consequences---involves devising unequivocal representations of components of behavior that are highly relevant to systematic responsibility attribution and to systematic blame-or-praise assignment. To put it plainly, expressive syntactic-and-semantic frameworks help us analyze responsibility-related problems in a methodical  way.\footnote{Most likely, this is why the logic-based formalization  of responsibility has become such an important topic in, for instance, normative multi-agent systems, responsible autonomous agents, and machine ethics for AI \cite{pereira2016programming,arrieta2020explainable}
}

The main goal of this paper is to present a proposal for a formal theory of responsibility. Such a proposal relies on (a) a \emph{decomposition} of responsibility into specific components and (b) a functional \emph{classification} of responsibility, where the different categories directly correlate with the components of the decomposition. As for the decomposition, it is given by the following list: 

\begin{itemize}
    \item[--]\label{itm:decom}\phantomsection\label{decom}\textbf{Agency}: the process by which agents bring about states of affairs in the environment. In other words, the phenomenon by which agents choose and perform actions, with accompanying mental states, that change the environment.
    \item[--]\textbf{Knowledge and}  \textbf{belief}:  mental states that concern the information available in the environment and that explain agents' particular choices of action.
    \item[--] \textbf{Intentions}: mental states that determine whether an action was done with the purpose of bringing about its effects. 
    \item[--] \textbf{Ought-to-do's}: the actions that agents should perform, complying to the codes of a normative system. Oughts-to-do's make up contexts that provide a criterion for deciding whether an agent should be blamed or praised. I refer to these contexts as the \emph{deontic contexts} of responsibility. 
\end{itemize} 

As for the classification, it is a refinement of Broersen's three categories of responsibility: \emph{causal, informational}, and \emph{motivational} responsibility \cite{aagotnes2006action, broersen2008complete, duijf2018let}. I will discuss these categories at length in Section~\ref{respo}. On the basis of both the decomposition and the classification, here I introduce a very rich stit logic to analyze responsibility, which I refer to as \emph{intentional epistemic act-utilitarian stit theory} (\emph{IEAUST}). More precisely, I use \emph{IEAUST} to model and syntactically characterize various modes of responsibility. By `modes of responsibility' I mean combinations of sub-categories of the three ones mentioned above, cast against the background of particular deontic contexts.  On the one hand, the sub-categories correspond to the different versions of responsibility that one can consider according to the \emph{active} and \emph{passive} forms of the notion: while the active form involves contributions---in terms of explicitly bringing about outcomes---the passive form involves omissions---which are interpreted as the processes by which agents allow that an outcome happens while being able, to some extent, to prevent it. On the other hand, the deontic context of a mode establishes whether and to what degree the combination of sub-categories involves either blameworthiness or praiseworthiness.

The logic \emph{IEAUST} includes a language that expresses agency, epistemic notions, intentionality, and different senses of obligation. With this language, I characterize the components of responsibility using particular formulas. Then, adopting a compositional approach---where complex modalities are built out of more basic ones---I use these characterizations of the components to formalize the aforementioned modes of responsibility. An outline of the paper is included below. 

\begin{itemize}
    \item Section~\ref{respo} presents an operational definition for responsibility and addresses the philosophical perspective adopted in my study of the notion.  
    
    \item Section~\ref{logic} introduces \emph{IEAUST} and uses this logic to provide stit-theoretic characterizations of different modes of responsibility.
    
    \item Section~\ref{axiomres} presents Hilbert-style proof systems both for \emph{IEAUST} and for a technical extension, addressing the status of their soundness \& completeness results.

\end{itemize}

\section{Categories of Responsibility}
\label{respo}

To make a start on formally analyzing responsibility, I identify (a) two \emph{viewpoints} for the philosophical study of responsibility, (b) three main \emph{categories} for the viewpoint that I focus on, and (c) two \emph{forms} in which the elements of the categories can be interpreted. 

As for point (a), the philosophical literature on responsibility usually distinguishes two \emph{viewpoints} on the notion \cite{van2011relation}: \emph{backward-looking responsibility} and \emph{forward-looking responsibility}. By backward-looking responsibility one refers to the viewpoint according to which an agent is considered to have produced a state of affairs that has already ensued and lies in the past. This is the viewpoint taken by a judge when, while trying a murder case, she wants to get to the bottom of things and find out who is responsible for doing the killing. In contrast, by forward-looking responsibility one refers to the viewpoint according to which which an agent is expected to comply with the duty of bringing about a state of affairs in the future. When one thinks of a student that has to write an essay before its due date, for instance, this is the view that is being used. In other words, the writing and the handing in of the essay before the deadline are seen as responsibilities of the student. 

From here on, I will focus on backward-looking responsibility. I work with the following operational definition: \phantomsection\label{elvin}\emph{responsibility} is a relation between the agents and the states of affairs of an environment, such that an agent is responsible for a state of affairs iff the agent's degree of involvement in the realization of that state of affairs warrants blame or praise (in light of a given normative system).   
As for point (b), I follow \cite{xstit} and \cite{duijf2018let} and distinguish three main \emph{categories} of responsibility, where each category can be correlated with the components of responsibility that it involves:\footnote{These categories extend the literature's common distinction between \emph{causal} and \emph{agentive} responsibility \cite{lorini2014logical, watson1996two,crisp2014aristotle}, and they were derived by \cite{xstit} on the basis of his analysis of the modes of \emph{mens rea}.}    

\begin{enumerate}
\item \textbf{\emph{Causal responsibility}}: an agent is causally responsible for a state of affairs iff the agent is the material author of such a state of affairs. The component that this category involves is agency. 
\item \label{itm:rcat2} \emph{\textbf{Informational responsibility}}: an agent is informationally responsible for a state of affairs iff the agent is the material author and it behaved knowingly, or consciously, while bringing about the state of affairs. The components that this category involves are agency, knowledge, and belief. 

\item \emph{\textbf{Motivational responsibility}}: an agent is motivationally responsible for a state of affairs iff the agent is the material author and it behaved knowingly and intentionally while bringing about the state of affairs. The components that this category involves are agency, knowledge, and intentions.
\end{enumerate}

Finally, as for point (c), the two \emph{forms} of responsibility are the \emph{active} form and the \emph{passive} form. The active form of responsibility concerns contributions, and the passive form of responsibility concerns 
omissions. 

Now, key elements in my operational definition of responsibility are the notions of blame and praise. Intuitively, responsibility can be measured by how much blame or how much praise an agent gets for its participation in bringing about a state of affairs.  As mentioned before,  \emph{ought-to-do's} can provide a criterion for deciding when agents should be blamed and when agents should be praised. The main idea is as follows: if agent $\alpha$ ought to have done $\phi$, then having seen to it that $\phi$ makes $\alpha$ praiseworthy, while having refrained from seeing to it that $\phi$ makes $\alpha$ blameworthy. For a given $\phi$, then, the degrees of $\alpha$'s praiseworthiness/blameworthiness correspond to the possible combinations between (a) an agent's ought-to-do's and (b) the active/passive forms of the three categories of responsibility. 

\section{A Logic of Responsibility}
\label{logic}

We are ready to introduce  \emph{intentional epistemic act-utilitarian stit theory} (\emph{IEAUST}), a stit-theoretic logic of responsibility. Without further ado, let me address the syntax and semantics of this expressive framework. 
\subsection{Syntax \& Semantics}

\begin{definition}[Syntax of intentional epistemic act-utilitarian stit theory]
\label{syntaxres}
Given a finite set $Ags$ of agent names and a countable set of propositions $P$, the grammar for the formal language $\mathcal L_{\textsf{R}}$ is given by
\[ \begin{array}{ll}
\phi ::= p \mid \neg \phi \mid \phi \wedge \phi \mid \Box \phi \mid [\alpha] \phi \mid K_\alpha \phi \mid I_\alpha \phi \mid \odot_\alpha \phi \mid\odot^{\mathcal{S}}_\alpha \phi,
\end{array} \]
where $p$ ranges over $P$ and $\alpha$ ranges over $Ags$. 
\end{definition}

In this language, $\square\varphi$ is meant to express the historical necessity of $\varphi$
($\Diamond \varphi$ abbreviates $\neg \square \neg \varphi$); $[\alpha] \varphi$ expresses that `agent $\alpha$ has seen to it that $\varphi$'; 
$K_\alpha \phi$ expresses that `$\alpha$ knew $\varphi$';
$I_\alpha \phi$ expresses that `$\alpha$ had a  present-directed intention toward the realization of $\varphi$';
$\odot_\alpha\phi$ expresses that `$\alpha$ objectively ought to have seen to it that $\phi$'; and $\odot^{\mathcal{S}}_\alpha \phi$ expresses that `$\alpha$ subjectively ought to have seen to it that $\phi$.' As for the semantics, the structures on which the formulas of $\mathcal L_{\mathtt{R}}$  are evaluated are based on what I call \emph{knowledge-intentions-oughts branching-time frames}. Let me first present the formal definition of these frames and then review the intuitions behind the extensions. 

\begin{definition}[\emph{Kiobt}-frames  \& models]
\label{kiobtframes}
A tuple $\left\langle M,\sqsubset, Ags,\mathbf{\mathbf{Choice}}, \left\{\sim_\alpha\right\}_{\alpha\in Ags}, \tau,  \mathbf{Value}  \right\rangle$ is called a \emph{knowledge-intention-oughts branching-time frame} (\emph{kiobt}-frame for short) iff
\begin{itemize}

\item $M$ is a non-empty set of \textnormal{moments} and $\sqsubset$ is a strict partial ordering on $M$ satisfying `no backward branching.' Each maximal $\sqsubset$-chain of moments is called a $\textnormal{history}$, where each history represents a complete temporal evolution of the world. $H$ denotes the set of all histories, and for each $m\in M$, $H_m:=\{h \in H ;m\in h\}$. Tuples $\left\langle m,h \right\rangle$ such that $m \in M$, $h \in H$, and $m\in h$, are called \emph{indices}, and the set of indices is denoted by $I(M\times H)$.  $\mathbf{Choice}$ is a function that maps each agent $\alpha$ and moment $m$ to a partition $\mathbf{Choice}^m_\alpha$ of $H_m$, where the cells of such a partition represent $\alpha$'s available actions at $m$. For $m\in M$ and $h\in H_m$, we denote the equivalence class of $h$ in $\mathbf{Choice}^m_\alpha$ by $\mathbf{Choice}^m_\alpha(h)$. $\mathbf{Choice}$ satisfies two constraints: \begin{itemize}
\item[$(\mathtt{NC})$] \emph{No choice between undivided histories}: For all $h, h'\in H_m$, if $m'\in h\cap h'$ for some $m' \sqsupset m$, then $h\in L$ iff $h'\in L$ for every $L\in \mathbf{Choice}^m_\alpha$. 

\item[$(\mathtt{IA})$]\emph{Independence of agency}: A function $s$ on $Ags$ is called a \emph{selection function} at $m$ if it assigns to each $\alpha$ a member of $\mathbf{Choice}^m_\alpha$. If we denote by $\mathbf{Select}^m$ the set of all selection functions at $m$, then we have that for every $m\in M$ and $s\in\mathbf{Select}^m$, $\bigcap_{\alpha \in Ags} s(\alpha)\neq \emptyset$ (see \cite{belnap01facing} for a discussion of the property).
\end{itemize}

\item For $\alpha\in Ags$, $\sim_\alpha$ is the epistemic indistinguishability equivalence relation for agent $\alpha$, which satisfies the following constraints:
\begin{itemize}
\item $(\mathtt{OAC})$ \emph{Own action condition}: if $\langle m_*, h_*\rangle\sim_\alpha \langle m, h\rangle$, then $\langle m_*,h_*'\rangle\sim_\alpha \langle m,h\rangle$ for every $h_*'\in \mathbf{Choice}^{m_*}_\alpha (h_*)$. We refer to this constraint as the `own action condition' because it implies that agents do not know more than what they perform.

\item $(\mathtt{Unif-H})$ \emph{Uniformity of historical possibility}: if $\langle m_*, h_*\rangle \sim_\alpha \langle m, h\rangle $, then for every $h_*'\in H_{m_*}$ there exists $h'\in H_m$ such that $\langle m_*, h_*'\rangle \sim_\alpha \langle m, h'\rangle$. Combined with $(\mathtt{OAC})$, this constraint is meant to capture a notion of uniformity of strategies, where epistemically indistinguishable indices should have the same available actions for the agent to choose upon.
\end{itemize}

For $\langle m,h\rangle$ and $\alpha\in Ags$, the set $\pi_\alpha^\square[\langle m, h \rangle]:=\{\langle m', h' \rangle; \exists h''\in H_{m'} s.t. \langle m, h \rangle\sim_\alpha\langle m', h'' \rangle \}$ is known as $\alpha$'s \emph{ex ante information set}.

 \item  $\tau$ is a function that assigns to each $\alpha\in Ags$ and index $\langle m, h \rangle$ a topology $\tau_\alpha^{\langle m, h \rangle}$ on $\pi_\alpha^\square\left[\langle m, h \rangle\right]$. This is the \emph{topology  of $\alpha$'s intentionality at $\langle m, h \rangle$}, where any non-empty open set is interpreted as a \emph{present-directed intention}, written `$\mbox{p-d}$ intention' from here on, of $\alpha$ at $\langle m, h\rangle$. Additionally, $\tau$ must satisfy the following conditions: 
    \begin{itemize}
    \item $(\mathtt{CI})$ \emph{Finitary consistency of intention}: for every $\alpha \in Ags$ and index $\langle m,h\rangle$, every non-empty $U, V\in\tau_\alpha^{\langle m, h \rangle}$ are such that $U\cap V\neq \emptyset$. In other words, every non-empty $U \in\tau_\alpha^{\langle m, h \rangle}$ is $\tau_\alpha^{\langle m, h \rangle}$-dense. 
    
    \item $(\mathtt{KI})$ \emph{Knowledge of intention}: for every $\alpha \in Ags$ and index $\langle m,h\rangle$, $\tau_\alpha^{\langle m, h \rangle}=\tau_\alpha^{\langle m', h' \rangle}$ for every $\langle m',h'\rangle$ such that $\pi_\alpha^\square\left[\langle m, h \rangle\right]=\pi_\alpha^\square\left[\langle m', h' \rangle\right]$. In other words, $\alpha$ has the same topology of $\mbox{p-d}$ intentions at all indices lying within $\alpha$'s current \emph{ex ante} information set.
  
\end{itemize}
\sloppy
\item $\mathbf{Value}$ is a deontic function that assigns to each history $h\in H$ a real number, representing the utility of $h$.

\end{itemize} 
A \emph{kiobt}-model $\mathcal{M}$, then, results from adding a valuation function $\mathcal{V}$ to a \emph{kiobt}-frame,  where $\mathcal{V}: P\to 
2^{I(M \times H)}$ assigns to each atomic proposition a set of indices.
\end{definition}

For $\alpha\in Ags$, the equivalence relation $\sim_\alpha$ is the usual indistinguishability relation, borrowed from epistemic logic, that represents $\alpha$'s uncertainty: whatever holds at all epistemically accessible indices is what an agent knows. As for the function $\tau$, it assigns to each agent the topology of  intentions, according to the ideas presented by \cite{abarca2022int}. The open sets of any such topology are taken to be $\mbox{p-d}$ intentions for bringing about circumstances. At each moment, the fact that the non-empty open sets of the topologies are dense implies that an agent's intentions are consistent. 

Regarding the deontic dimension, the idea is that objective, subjective, and doxastic ought-to-do's stem from the optimal actions for an agent: to have seen to it that $\phi$ is taken to be an obligation of an agent at an index iff $\phi$ is an effect of all the optimal actions for that agent and index, where the notion of optimality is based on the deontic value of the histories in those actions---provided by $\mathbf{Value}$.  The semantics for formulas involving the deontic operators require some previous definitions. For $m\in M$ and $\beta\in Ags$, we define $\mathbf{State}_\beta^m=\left\{S\subseteq H_m; S=\bigcap_{\alpha \in Ags-\{\beta\}} s(\alpha), \mbox{ where }s\in \mathbf{Select}^m\right\}$. For $\alpha\in Ags$ and $m_*\in M$, we first define
a general ordering $\leq$ on $\mathcal{P}(H_{m_*})$ such that for $X, Y\subseteq H_{m_*}$, $
X\leq Y \textnormal{ iff } \mathbf{Value}(h) \leq \mathbf{Value}(h') \textnormal{ for every } h\in X, h'\in Y$. The \emph{objective} dominance ordering $\preceq$ is defined such that for $L, L'\in \mathbf{Choice}_\alpha^{m_*}$, $
L\preceq L' \textnormal{ iff }  \mbox{for each } S\in \mathbf{State}_\alpha^{m_*}, L\cap S \leq L'\cap S.$ The optimal set of actions is taken as $\mathbf{Optimal}^{m_*}_\alpha:=\{L \in \mathbf{Choice}^{m_*}_\alpha ; \textnormal{there is no } L' \in \mathbf{Choice}^{m_*}_\alpha \textnormal{such that }  L\prec L'\}.$    

\emph{Subjective} ought-to-do's involve a different dominance ordering. To define it, \cite{JANANDI} and \cite{abarca2019logic} introduce the so-called \emph{epistemic clusters}, which are nothing more than a given action's epistemic equivalents in indices that are indistinguishable to the one of evaluation. Formally, we have that for  $\alpha\in Ags$, $m_*, m\in M$, and $L\subseteq H_{m_*}$, $L$'s \emph{epistemic cluster} at $m$ is the set $[L]^m_\alpha:=\{h\in H_m ; \exists h_*\in L \ \textnormal{ s.t. }\ \langle m_*,h_*\rangle \sim_\alpha \langle m,h\rangle \}$. A subjective dominance ordering $\preceq_s$ on $\mathbf{Choice}^{m_*}_\alpha$ is then defined by the following rule: for $L, L'\subseteq H_{m_*}$, $L\preceq_s L'$ iff $\mbox{for each } m$ such that $m_*\sim_\alpha m$, $\mbox{for each } S\in \mathbf{State}_\alpha^{m}, [L]^m_\alpha\cap S \leq [L']^m_\alpha\cap S$.\footnote{As a convention, I write $m \sim_\alpha m'$ if there exist $h\in H_m$, $h'\in H_{m'}$ such that $\langle m,h\rangle \sim_\alpha \langle m',h'\rangle$.} Just as in the case of objective ought-to-do's, this ordering allows us to define a subjectively optimal set of actions $\mathbf{SOptimal}^{m_*}_\alpha:=\{L \in \mathbf{Choice}^{m_*}_\alpha ; \textnormal{ there is no } L' \in \mathbf{Choice}^{m_*}_\alpha \textnormal{ s. t. }  L\prec_s L'\},$ where I write $L\prec_s L'$ iff $L\preceq_s L'$ and $L'\npreceq_s L$.

\fussy

Therefore, \emph{kiobt}-frames allow us to represent the  components of responsibility discussed in the introduction: agency, knowledge, intentions, and ought-to-do's. More precisely, they allow us to provide semantics for the modalities of $\mathcal{L}_{\textsf{R}}$:

\begin{definition}[Evaluation rules for \emph{IEAUST}]
\label{evares}
Let $\mathcal{M}$ be a finite-choice \emph{kiobt}-model.\footnote{Finite-choice \emph{bt}-models are those for which function $\mathbf{Choice}$ is such that $\mathbf{Choice}_\alpha^m$ is finite for every $\alpha\in Ags$ and $m\in M$. I focus on finite-choice models to simplify the evaluation rules for objective and subjective ought-to-do's. The reader is referred to \cite{abarca2019logic} for the evaluation rules in the case of infinite-choice models.}  The semantics on $\mathcal{M}$ for the formulas of $\mathcal {L}_{\textsf{R}}$ are recursively defined by the following truth conditions: 
\[ \begin{array}{lll}
\mathcal{M},\langle m,h \rangle \models p & \mbox{ iff } & \langle m,h \rangle \in \mathcal{V}(p) \\

\mathcal{M},\langle m,h \rangle \models \neg \phi & \mbox{ iff } & \mathcal{M},\langle m,h \rangle \not\models \phi \\

\mathcal{M},\langle m,h \rangle \models \phi \wedge \psi & \mbox{ iff } & \mathcal{M},\langle m,h \rangle \models \phi \mbox{ and } \mathcal{M},\langle m,h \rangle \models \psi \\

\mathcal{M},\langle m,h \rangle \models \Box \phi &
\mbox{ iff } & \mbox{for all } h'\in H_m, \mathcal{M},\langle m,h' \rangle \models \phi \\

\mathcal{M},\langle m,h \rangle \models [\alpha]
\phi & \mbox{ iff } & \mbox{for all } h'\in \mathbf{Choice}^m_\alpha(h), \mathcal{M},\langle m, h'\rangle \models \phi\\

\mathcal{M},\langle m,h \rangle \models K_{\alpha} \phi &
\mbox{ iff } & \mbox{for all } \langle m',h'\rangle \mbox{ s. t. }  \langle m,h \rangle \sim_{\alpha}\langle m',h' \rangle,\\&&  \mathcal{M},\langle m',h' \rangle \models \phi\\

\mathcal{M},\langle m,h \rangle \models I_\alpha\phi &
\mbox{ iff } & \mbox{there exists } U\in \tau_\alpha^{\langle m, h \rangle} \mbox{ s. t. }  U\subseteq \|\phi \|\\

\mathcal{M},\langle m,h \rangle \models \odot_\alpha \phi &
\mbox{ iff } & \mbox{for all } L\in \mathbf{Optimal} ^{m}_\alpha,  \mathcal{M},\left\langle m,h'\right\rangle\models \varphi \\&& \mbox{for every } h'\in L\\

\mathcal{M},\langle m,h \rangle \models \odot^{\mathcal{S}}_\alpha \varphi & \mbox{ iff }  & \mbox{for all } L\in \mathbf{SOptimal} ^{m}_\alpha,  \mathcal{M},\left\langle m',h'\right\rangle\models \varphi \\&& \mbox{for every } m'  \mbox{ s. t. } m\sim_\alpha m'\mbox{ and every } h'\in [L]^{m''}_\alpha.\\
\end{array} \]
where $\|\phi\|$ refers to the set $\left\{\left\langle m,h\right\rangle \in I(M\times H) ;\mathcal{M},\left\langle m, h\right\rangle \models \phi\right\}$. 
\end{definition}

\subsection{Formalization of Sub-Categories of Responsibility}
\label{form}

The logic introduced in the previous subsection allows us to formalize different modes of responsibility by means of formulas of $\mathcal{L}_{\textsf{R}}$. Before diving into the formulas, let me present an operational definition for the expression `mode of responsibility.' For $\alpha\in Ags$, index $\langle m, h \rangle$, and $\phi$ of $\mathcal{L}_{\textsf{R}}$, a \emph{mode of $\alpha$'s responsibility with respect to $\phi$ at $\langle m, h \rangle$} is a tuple consisting of three constituents: (1) a set of categories, taken from Broersen's three categories of responsibility, that applies to the relation between $\alpha$ and $\phi$ at $\langle m, h \rangle$, (2) the forms of responsibility---active or passive---that apply to the categories in said set, and (3) a deontic context, determining whether the forms of the categories are either blameworthy, praiseworthy, or neutral. As for constituents (1) and (2), observe that the active and passive forms of the three categories of responsibility lead to sub-categories of the notion. For clarity, first I will introduce the stit-theoretic characterizations of these sub-categories; afterwards, in Subsection~\ref{formmod}, these sub-categories will be discussed against the backdrop of the deontic contexts that will decide their degree of blameworthiness or praiseworthiness (constituent (3) in a given mode). 

A maxim usually endorsed in the philosophical literature on moral responsibility is the \emph{principle of alternate possibilities}. According to this principle, ``a person is morally responsible for what he has done only if he could have done otherwise'' \cite{frankfurt2018alternate}.
Following the example of \cite{lorini2014logical}, then, I adopt the intuitions behind deliberative agency and restrict my view on responsibility to situations where agents can be said to actually have had a hand in bringing about  states of affairs. Therefore, each sub-category of $\alpha$'s responsibility with respect to $\phi$ at $\langle m, h\rangle$ will include a positive condition---concerning the realization of $\phi$---and a negative condition---concerning the realization of $\lnot \phi$.  For $\alpha\in Ags$ and $\phi$ of $\mathcal{L}_{\textsf{R}}$, the main sub-categories of $\alpha$'s responsibility with respect to $\phi$ are displayed in Table~\ref{table:kb0}.
\begin{table}[!htb]
\centering\renewcommand\cellalign{lc}
\setcellgapes{4pt}\makegapedcells
\begin{tabularx}{.95\textwidth}{ |>{\centering\arraybackslash}X | >{\centering\arraybackslash}X |>{\centering\arraybackslash}X | } \hline\diagbox{\emph{Category}}{\emph{Form}} & Active (contributions) &  Passive (omissions)\\ \hline \makecell{Causal} & \makecell{$[\alpha]\phi\land \Diamond[\alpha] \lnot\phi$} & \makecell{$\phi \land \Diamond [\alpha]\lnot \phi$} \\ \hline \makecell{Informational} & \makecell{$K_\alpha[\alpha]\phi\land \Diamond K_\alpha [\alpha] \lnot\phi$ } & \makecell{$\phi\land K_\alpha\lnot[\alpha]\lnot\phi \land$\\ $\Diamond K_\alpha[\alpha]\lnot \phi$}\\ \hline \makecell{Motivational} & \makecell{$K_\alpha[\alpha]\phi\land I_\alpha[\alpha]\phi\land$ \\ $\Diamond K_\alpha [\alpha]\lnot\phi$} & \makecell{$\phi\land K_\alpha\lnot[\alpha]\lnot\phi\land$\\$I_\alpha\lnot[\alpha]\lnot\phi\land\Diamond K_\alpha[\alpha]\lnot \phi$}\\ \hline\end{tabularx}
\caption{Main sub-categories.}
\label{table:kb0}
\end{table}

Let me explain and discuss Table~\ref{table:kb0}. Let $\mathcal{M}$ be a \emph{kiobt}-model. For $\alpha\in Ags$ and index $\langle m,h\rangle$, the sub-categories of $\alpha$'s responsibility with respect to $\phi$ at $\langle m, h \rangle$ are defined as follows: \begin{itemize}
        \item\phantomsection\label{gogol} $\alpha$ was \emph{causal-active responsible} for $\phi$ at  $\langle m,h\rangle$  iff  at $\langle m,h\rangle$ $\alpha$ has seen to it that $\phi$ (the positive condition) and it was possible for $\alpha$ to prevent $\phi$ (the negative condition). As such, I refer to state of affairs $\phi$ as a causal contribution of $\alpha$ at $\langle m,h\rangle$. $\alpha$  was \emph{causal-passive responsible} for $\phi$ at $\langle m,h\rangle$  iff at $\langle m,h\rangle$   $\phi$ was the case (the positive condition), and $\alpha$ refrained from preventing $\phi$ while it was possible for $\alpha$ to prevent $\phi$ (the negative conditions). To clarify, formula $\phi\to \lnot[\alpha]\lnot\phi$ is valid, so that if $\phi$ was the case then $\alpha$ refrained from preventing $\phi$. I refer to $\lnot\phi$ as a causal omission of $\alpha$ at $\langle m,h\rangle$. 
    \item $\alpha$  was \emph{informational-active responsible} for $\phi$ at $\langle m,h\rangle$  iff at $\langle m,h\rangle$ $\alpha$ has knowingly seen to it that $\phi$ (the positive condition) and it was possible for $\alpha$ to knowingly prevent $\phi$ (the negative condition). I refer to $\phi$ as a conscious contribution of $\alpha$ at $\langle m,h\rangle$. $\alpha$ was \emph{informational-passive responsible} for $\phi$ at $\langle m,h\rangle$  iff at $\langle m,h\rangle$ $\phi$ was the case (the positive condition), and $\alpha$ knowingly refrained from preventing $\phi$ while it was possible for $\alpha$ to knowingly prevent $\phi$ (the negative conditions). I refer to $\lnot\phi$ as a conscious omission of $\alpha$ at $\langle m,h\rangle$.
 
    \item $\alpha$ was \emph{motivational-active responsible} for $\phi$ at $\langle m,h\rangle$  iff at $\langle m,h\rangle$  $\alpha$ has both knowingly and intentionally seen to it that $\phi$ (the positive conditions) and it was possible for $\alpha$ to knowingly prevent $\phi$ (the negative condition). I refer to $\phi$ as a motivational contribution of $\alpha$ at $\langle m,h\rangle$.
         $\alpha$ was \emph{motivational-passive responsible} for $\phi$ at $\langle m,h\rangle$  iff at $\langle m,h\rangle$  $\phi$ was the case (the positive condition), and $\alpha$ both knowingly and intentionally refrained from preventing $\phi$  while it was possible for $\alpha$ to knowingly prevent $\phi$ (the negative conditions). I refer to $\lnot\phi$ as a motivational omission of $\alpha$ at $\langle m,h\rangle$.

\end{itemize}  
  
The main reason for setting the negative conditions as stated in Table~\ref{table:kb0} is that it greatly simplifies the relation between the active and the passive forms of responsibility. That said, it is important to mention that these negative conditions lead to a policy that I call \emph{leniency on blameworthy agents}. 

Two important observations concerning the relations between these sub-categories are the following: 
\begin{enumerate}
    \item \begin{enumerate}[(a)]
    \item If $\alpha$ was informational-active, resp. informational-passive, responsible for $\phi$ at $\langle m, h \rangle$, then $\alpha$ was causal-active, resp. causal-passive, responsible for $\phi$ at $\langle m, h \rangle$; the converse is not true. 
    
    \item If $\alpha$ was motivational-active, resp. motivational-passive, responsible for $\phi$ at $\langle m, h \rangle$, then $\alpha$ was informational-active, resp. informational-passive, responsible for $\phi$ at $\langle m, h \rangle$; the converse is not true. 
    \end{enumerate}
    
    \item \label{itm:bside} For all three categories, the active form of  responsibility with respect to $\phi$ implies the passive form.
\end{enumerate}

\subsection{Formalization of Modes of Responsibility}
\label{formmod}

In Section~\ref{respo} I explained that obligations provide the deontic contexts of responsibility, which in turn determine degrees of praiseworthiness/blameworthiness for instances of the notion. Let $\mathcal{M}$ be a \emph{kiobt}-model. Take $\alpha\in Ags$, and let $\phi$ be a formula of $\mathcal{L}_{\textsf{R}}$. For each index $\langle m, h\rangle$, there are 4 main possibilities for conjunctions of deontic modalities holding at $\langle m, h\rangle$, according to whether $\Delta\phi$ or $\lnot \Delta\phi$ is satisfied at the index, where $\Delta\in\left\{\odot_\alpha, \odot^{\mathcal{S}}_\alpha\right\}$. I refer to any such conjunction as a \emph{deontic context for $\alpha$'s responsibility with respect to $\phi$ at $\langle m, h\rangle$}. Thus, these contexts render 4 main levels of praiseworthiness, resp. blameworthiness, under the premise that bringing about $\phi$ is praiseworthy and refraining from bringing about $\phi$ is blameworthy. I use numbers $1$--$4$ to refer to these levels, so that \emph{Level} $1$ corresponds the highest level of praiseworthiness, resp. blameworthiness, and \emph{Level} $4$ corresponds to the lowest level. 

\noindent\underline{\textbf{\emph{Level 1}}}: when deontic context $\odot_\alpha\phi\land \odot^{\mathcal{S}}_\alpha\phi$ holds at $\langle m, h\rangle$, which occurs iff at $\langle m, h\rangle$ $\alpha$ objectively and subjectively ought to have seen to it that $\phi$. \noindent\underline{\textbf{\emph{Level 2}}}: when deontic context $\lnot \odot_\alpha\phi\land \odot^{\mathcal{S}}_\alpha\phi$ holds at $\langle m, h \rangle$, which occurs iff at $\langle m, h\rangle$ $\alpha$ subjectively ought to have seen to it that $\phi$, but $\alpha$ did not objectively ought to have seen to it that $\phi$. \noindent\underline{\textbf{\emph{Level 3}}}: when deontic context  $ \odot_\alpha\phi\land \lnot \odot^{\mathcal{S}}_\alpha\phi$ holds at $\langle m, h \rangle$, which occurs iff at $\langle m, h\rangle$ $\alpha$ objectively ought to have seen to it that $\phi$, but $\alpha$ did not
subjectively ought to have seen to it that $\phi$. \noindent\underline{\textbf{\emph{Level 4}}}: when deontic context  $\lnot\odot_\alpha\phi\land  \lnot \odot^{\mathcal{S}}_\alpha\phi$ holds at $\langle m, h \rangle$, where, unless $\alpha$ either objectively or subjectively ought have seen to it that $\lnot \phi$ at $\langle m, h \rangle$ (which would imply that a deontic context of the previous levels holds with respect to $\lnot\phi$), neither bringing about $\phi$ nor refraining from doing so elicits any interest in terms of blame-or-praise assignment. 

For each of these deontic contexts, the \emph{basic modes of $\alpha$'s active responsibility with respect to $\phi$ at $\langle m, h\rangle$} are displayed in Table~\ref{table:modes0}, and the \emph{basic modes of $\alpha$'s passive responsibility} are obtained by substituting the term 'passive' for 'active' in such a table.
 \begin{table}[!htb]
\centering\renewcommand\cellalign{lc}
\setcellgapes{3pt}\makegapedcells
\begin{tabularx}{\textwidth}{ |>{\centering\arraybackslash\hsize=.7\hsize}X | >{\centering\arraybackslash\hsize=1.2\hsize}X |>{\centering\arraybackslash\hsize=1.1\hsize}X | } \hline\diagbox{\textbf{\emph{Deg.}}}{\textbf{\emph{Att.}}} & \makecell{\textbf{Praiseworthiness}} &  \makecell{\textbf{Blameworthiness}} \\ \hline \makecell{$\mbox{\emph{Low}}_A$} & \makecell{\small{Causal-active for $\phi$ \checkmark}\\ \hline \small{Infor.-active for $\phi$ \xmark}\\ \hline \small{Motiv.-active for $\phi$ \xmark}} & \makecell{\small{Causal-active for $\lnot\phi$ \checkmark}\\ \hline \small{Infor.-active for $\lnot\phi$ \xmark}\\ \hline \small{Motiv.-active for $\lnot\phi$ \xmark}} \\ \hline   \makecell{$\mbox{\emph{Middle}}_A$} & \makecell{\small{Causal-active  for $\phi$ \checkmark}\\ \hline \small{Infor.-active  for $\phi$ \checkmark}\\ \hline \small{Motiv.-active  for $\phi$ \xmark}} & \makecell{\small{Causal-active for $\lnot\phi$ \checkmark}\\ \hline \small{Infor.-active for $\lnot\phi$ \checkmark}\\ \hline \small{Motiv.-active for $\lnot\phi$ \xmark}}   \\ \hline \makecell{$\mbox{\emph{High}}_A$} & \makecell{\small{Causal-active for $\phi$ \checkmark}\\ \hline \small{Infor.-active for $\phi$ \checkmark}\\ \hline \small{Motiv.-active for $\phi$ \checkmark}} & \makecell{\small{Causal-active for $\lnot\phi$ \checkmark}\\ \hline \small{Infor.-active  for $\lnot\phi$ \checkmark}\\ \hline \small{Motiv.-active for $\lnot\phi$ \checkmark}}  \\ \hline \end{tabularx}
\caption{Modes of $\alpha$'s active responsibility with respect to $\phi$.}
\label{table:modes0}
\end{table}

\section{Axiomatization}
\label{axiomres}

This section is devoted to introducing proof systems for \emph{IEAUST}. More precisely, I present two systems:

\begin{itemize}
    \item A sound system for \emph{IEAUST}, for which achieving a completeness result is still an open problem.
    \item A sound and complete system for a technical extension of \emph{IEAUST} that I refer to as \emph{bi-valued} \emph{IEAUST}. Bi-valued \emph{IEAUST} was devised with the aim of having a completeness result for a logic that would be reasonably similar to the one presented in Section~\ref{logic}. 
\end{itemize}

As for the first bullet point, a proof system for  \emph{IEAUST} is defined as follows:

\begin{definition}[Proof system for  \emph{IEAUST}]
\label{axiomres1}
Let $\Lambda_R$ be the proof system defined by the following axioms and rules of inference: 
\begin{itemize}
\item \emph{(Axioms)} All classical tautologies from propositional logic; the $\mathbf{S5}$ schemata for $\square$, $[\alpha]$,  and $K_\alpha$; the $\mathbf{KD}$ schemata for $I_\alpha$; and the schemata given in Table~\ref{table:magicjohnson}.
\begin{table}[!htb]
\begin{tabularx}{\textwidth}{|>{\raggedleft\arraybackslash}X | >{\raggedleft\arraybackslash}X|} \hline \makecell[t]{\emph{Basic-stit-theory schemata}:\\
\small
$\begin{array}{ll}
\square \phi\to [\alpha ] \phi & (SET)\\
\mbox{For distinct $\alpha_1,\dots,\alpha_m$}, &\\
\bigwedge\limits_{1\leq k\leq m}\Diamond [\alpha_i ] \phi_i \to \Diamond\left(\bigwedge\limits_{1\leq k\leq m}[\alpha_i ] \phi_i\right)& (IA)
\end{array}$}
& \makecell[t]{\emph{Schemata for knowledge:}\\
\small
$\begin{array}{ll}
K_\alpha \phi\to [\alpha ]\phi&(OAC)\\
\Diamond K_\alpha \phi \to K_\alpha \Diamond  \phi&(Unif-H)
\end{array}$} \\ \hline \makecell[t]{\emph{Schemata for objective ought-to-do's:}\\
\small
$\begin{array}{ll}
\odot_\alpha (\phi\to \psi)\to (\odot_\alpha \phi \to \odot_\alpha \psi)& (A1)\\ 
\square \phi\to \odot_\alpha \phi& (A2)\\
\odot_\alpha \phi\to\square\odot_\alpha \phi&(A3)\\ 
\odot_\alpha \phi\to  \odot_\alpha ([\alpha ]\phi)&(A4)\\
\odot_\alpha \phi\to \Diamond  [\alpha ] \phi &(Oic)
\end{array}$} & \makecell[t]{\emph{Schemata for subjective ought-to-do's:}\\
\small
$\begin{array}{ll}
\odot^{\mathcal{S}}_\alpha (\phi\to \psi)\to (\odot^{\mathcal{S}}_\alpha \phi \to \odot^{\mathcal{S}}_\alpha \psi)& (A5)\\
\odot^{\mathcal{S}}_\alpha \phi\to  \odot^{\mathcal{S}}_\alpha (K_\alpha \phi)&(A6)\\
K_\alpha \square \phi\to\odot^{\mathcal{S}}_\alpha\phi &(SuN)\\ \odot^{\mathcal{S}}_\alpha\phi\to\Diamond K_\alpha \phi & (s.Oic)\\
\odot^{\mathcal{S}}_\alpha \phi \to K_{\alpha} \square\odot^{\mathcal{S}}_\alpha \phi
&(s.Cl)\\
\odot^{\mathcal{S}}_\alpha \phi \to \lnot \odot_\alpha \lnot\phi
&(ConSO)
\end{array}$}\\ \hline  \makecell[t]{\emph{Schemata for intentionality:}\\
\small
$\begin{array}{ll}
\square  K_\alpha \phi \to  I_\alpha\phi &(InN)\\
I_\alpha\phi\to \square  K_\alpha I_\alpha \phi &(KI)
\end{array}$} & \\ \hline\end{tabularx}
\caption{Axioms for the modalities' interactions.}
\label{table:magicjohnson}
\end{table}

\item \textit{(Rules of inference)} \textit{Modus Ponens}, Substitution, and Necessitation for all modal operators.
\end{itemize}
\end{definition}

For a discussion of all these axioms and schemas, the reader is referred to \cite{abarca2019logic, MUR, abarca2022int}. 
An important result for $\Lambda_R$, then, is the following proposition, whose proof is relegated to Appendix \ref{metalogicres}. 

\begin{restatable}[Soundness of $\Lambda_R$]{proposition}{soundres}
The proof system ${\Lambda_R}$ is sound with respect to the class of \emph{kiobt}-models.
\end{restatable}

Unfortunately, the question of whether $\Lambda_R$ is complete with respect to the class of \emph{kiobt}-models is still an open problem. Now, in the search for a complete proof system for \emph{IEAUST}, and following a strategy found in my joint works with Jan Broersen \cite{abarca2019logic, abarca2021deontic}, I tried to first prove completeness of $\Lambda_R$ with respect to a class of more general models, that I refer to as \emph{bi-valued} \emph{kiobt}-models (Definition~\ref{multikb} below). This strategy led to the need of dropping one of the schemata in $\Lambda_R$: $(ConSO)$. More precisely, if $\Lambda_{R}'$ is obtained from $\Lambda_R$ by eliminating $(ConSO)$ in Definition~\ref{axiomres1}, then $\Lambda_{R}'$ turns out to be sound and complete with respect to the class of  \emph{bi-valued} \emph{kiobt}-models. The formal statements are included below.

\begin{definition}[Bi-valued \emph{kiobt}-frames  \& models]
\label{multikb}
$\left\langle M,\sqsubset, Ags,\mathbf{\mathbf{Choice}}, \left\{\sim_\alpha\right\}_{\alpha\in Ags}, \tau,  \mathbf{Value}_{\mathcal{O}}, \mathbf{Value}_{\mathcal{S}} \right\rangle$ is called a \emph{bi-valued} \emph{kiobt}-frame iff
\begin{itemize}

\item $M, \sqsubset, Ags, \mathbf{\mathbf{Choice}}$, $\left\{\sim_\alpha\right\}_{\alpha\in Ags}$, and $\tau$ are defined just as in Definition~\ref{kiobtframes}.

\item $\mathbf{Value}_{\mathcal{O}}$ and $\mathbf{Value}_{\mathcal{S}}$ are functions that independently assign to each history $h\in H$ a real number. 

\end{itemize} 
A \emph{bi-valued kiobt}-model $\mathcal{M}$, then, results from adding a valuation function $\mathcal{V}$ to a bi-valued \emph{kiobt}-frame,  where $\mathcal{V}: P\to 
2^{I(M \times H)}$ assigns to each atomic proposition of $\mathcal{L}_{\textsf{R}}$ a set of indices (recall that $P$ is the set of propositions in $\mathcal{L}_{\textsf{R}}$).
\end{definition}

The two value functions in bi-valued \emph{kiobt}-frames allow us to redefine the dominance orderings so that they are independent from one another, something that proves useful in achieving a completeness result in the style of \cite{abarca2019logic}. For $\alpha\in Ags$ and $m\in M$, two general orderings $\leq$ and $\leq_s$ are first defined on $2^{H_{m}}$: for $X, Y\subseteq H_{m}$, $
X\leq Y$, resp. $
X\leq_s Y$,  iff  $\mathtt{Value}_\mathcal{O}(h) \leq \mathtt{Value}_\mathcal{O}(h')$, resp. $\mathtt{Value}_\mathcal{S}(h) \leq \mathtt{Value}_\mathcal{S}(h')$, for every $h\in X$ and  $h'\in Y$. Then, for $\alpha\in Ags$ and $m\in M$, an objective dominance ordering $\preceq$ is now defined on $\mathbf{Choice}_\alpha^{m}$ by the rule: $
L\preceq L'$ iff for every  $S\in \mathbf{State}_\alpha^{m}, L\cap S \leq L'\cap S.$ 
 In turn, for $\alpha\in Ags$ and $m\in M$, a subjective dominance ordering $\preceq_s$ is now  defined on $\mathbf{Choice}_\alpha^{m}$ by the rule: $L\preceq_s L'$ iff for all $m'$ such that $m\sim_\alpha m'$ and each $S\in \mathbf{State}_\alpha^{m}, [L]^{m'}_\alpha\cap S \leq_s [L']^{m'}_\alpha\cap S$. With these new notions, the sets  $\mathbf{Optimal}^m_\alpha$ and $\mathbf{SOptimal}^m_\alpha$ are redefined accordingly, and the evaluation rules for the formulas of $\mathcal{L}_{\textsf{R}}$ (with respect to bi-valued \emph{kiobt}-models) are given just as in Definition~\ref{evares}. As mentioned before, I refer to the resulting logic as \emph{bi-valued} \emph{IEAUST}. Bi-valued \emph{IEAUST}, then, admits the following metalogic result, whose proof is sketched in Appendix \ref{metalogicres}.

\begin{restatable}[Soundness \& Completeness of $\Lambda_R'$]{theorem}{soundcompres}
Let $\Lambda_{R}'$ be the proof system  obtained from $\Lambda_{R}$ by eliminating $(ConSO)$ in Definition~\ref{axiomres1}. Then $\Lambda_{R}'$ is sound and complete with respect to the class of bi-valued \emph{kiobt}-models.
\end{restatable}

\section{Conclusion}

This paper built a formal theory of responsibility by means of stit-theoretic models and languages that were designed to explore the interplay between the following components of responsibility: agency, knowledge, beliefs, intentions, and obligations. Said models were integrated into a framework that is rich enough to provide logic-based characterizations for different instances of three categories of responsibility: causal, informational, and motivational responsibility. 

The developed theory belongs to a relatively recent tradition in the philosophical literature, that seeks to formalize responsibility allocation by means of models of agency and logic-based languages (see, for instance, \cite{de2010logic}, \cite{lorini2014logical}, \cite{alechina2017causality}, \cite{naumov2019blameworthiness}, \cite{naumov2020epistemic}, and \cite{baier2021game}). Most of these frameworks characterize different forms of responsibility as combinations of causal agency, knowledge, and the principle of alternate possibilities. The novelty of the present approach, then, lies in the introduction of intentionality and ought-to-do's. Such an introduction gives rise to a taxonomy that distinguishes various kinds of responsibility and blameworthiness/praiseworthiness in a methodical, meticulous way. Interesting directions for future work, then, involve extending these models with beliefs and rational decision-making, group notions (coalitions, group knowledge \& belief,  collective intentionality, collective responsibility), temporal modalities, and long-term strategies, for instance. As for the technical aspects of the formal theory, an important directions for future work involve checking whether the logic is decidable, checking for the complexity of its satisfiability problem, and figuring out its applicability for implementation.\footnote{Implementing logics of responsibility might prove relevant in the design, formal verification, and explainability of ethical AI (see, for instance, \cite{calegari2020logic}).}

\bibliographystyle{eptcs}
\bibliography{generic}
\begin{subappendices}

\setcounter{section}{0}
\renewcommand{\thesection}{\Alph{section}}%
\setcounter{theorem}{0}
\renewcommand{\thetheorem}{\thesection.\arabic{theorem}}
\renewcommand{\theproposition}{\thesection.\arabic{theorem}}
\renewcommand{\thelemma}{\thesection.\arabic{theorem}}
\renewcommand{\thedefinition}{\thesection.\arabic{theorem}}
\renewcommand{\theobservation}{\thesection.\arabic{theorem}}
\renewcommand{\theremark}{\thesection.\arabic{theorem}}
\renewcommand{\thecorollary}{\thesection.\arabic{theorem}}

\section{Metalogic Results for \emph{IEAUST}}\label{metalogicres}

\subsection{Soundness}

\begin{proposition}[Soundness of ${\Lambda_R}$] \label{soundres} The system ${\Lambda_R}$ (Definition~\ref{axiomres}) is sound with respect to the class of \emph{kiobt}-models. 
\end{proposition}

\begin{proof}
The proof of soundness is routine: the validity of the $\mathbf{S5}$ schemata for $\square$ and $[\alpha]$, as well as that of $(SET)$ and $(IA)$, is standard from \cite{xu1994decidability}; the validity of the $\mathbf{S5}$ schemata for $K_\alpha$ is standard from epistemic logic; the validity of schemata $(A1)$--$(A4)$, as well as that of $(Oic)$, is standard from \cite{MUR}; the validity of the $\mathbf{KD}$ schemata for $I_\alpha$, as well as that of $(InN)$, follows from Definitions~\ref{kiobtframes} and \ref{evares}; and the validity of $(KI)$ follows from frame condition $(\mathtt{KI})$; and the validity of schemata $(OAC)$, $(Unif-H)$, $(A5)$ and $(A6)$, as well as that of $(SuN)$, $(s.Oic)$, $(s.Cl)$, and $(ConSO)$ can be shown as follows:

\item To see that $\mathcal{M}\models (OAC)$, take $\left\langle m,h\right\rangle$ such that $\mathcal{M},\left\langle m,h\right\rangle\models K_\alpha \varphi$. Take $h'\in \mathbf{Choice}_\alpha^{m}(h)$. Frame condition $(\mathtt{OAC})$ implies that $\left\langle m,h\right\rangle \sim_\alpha \left\langle m,h'\right\rangle$. The assumption that  $\mathcal{M},\left\langle m,h\right\rangle\models K_\alpha \varphi$ then implies that $\mathcal{M},\left\langle m,h'\right\rangle\models \varphi$. Therefore, for any $h'\in \mathbf{Choice}_\alpha^{m}(h)$, $\mathcal{M}, \left\langle m,h'\right\rangle\models \varphi$, which implies that $\mathcal{M},\left\langle m,h\right\rangle\models [\alpha ] \varphi$.

\item To see that $\mathcal{M}\models (Unif-H)$, take $\left\langle m,h\right\rangle$ such that $\mathcal{M},\left\langle m,h\right\rangle\models \Diamond K_\alpha \varphi$. Let $\left\langle m',h'\right\rangle$  be an index such that $\left\langle m,h\right\rangle\sim_\alpha \left\langle m',h'\right\rangle$. We want to show that $\mathcal{M},\left\langle m',h'\right\rangle\models  \Diamond\varphi$.  The fact that $\mathcal{M},\left\langle m,h\right\rangle\models \Diamond K_\alpha \varphi$ implies that there exists $h_*\in H_m$ such that $(\star)$ $\mathcal{M},\left\langle m,h_*\right\rangle\models K_\alpha \varphi$. Frame condition $(\mathtt{Unif-H})$ implies that there exists $h_*'\in H_{m'}$ such that $\left\langle m,h_*\right\rangle\sim_\alpha \left\langle m',h_*'\right\rangle$. With $(\star)$, this last fact implies that $\mathcal{M},\left\langle m',h_*'\right\rangle\models  \varphi$, which in turn implies that $\mathcal{M},\left\langle m',h'\right\rangle\models  \Diamond\varphi$. Therefore, $\mathcal{M},\left\langle m,h\right\rangle\models K_\alpha \Diamond\varphi$.

\item To see that $\mathcal{M}\models (A6)$, take $\left\langle m,h\right\rangle$ such that $\mathcal{M},\left\langle m,h\right\rangle\models \odot^{\mathcal{S}}_\alpha  \varphi$. We want to show that, for every $L\in \mathbf{Choice}^{m}_\alpha$ such that $[L]^{m'}\not\subseteq\left|K_\alpha\varphi\right|^{m'}$ (for some $m'$ such that $m\sim_\alpha m'$), there is $L'\in \mathbf{Choice}_\alpha^{m}$ such that $L\prec_s L'$ and, if  $L''=L'$ or $L'\preceq_s L''$, then  $[L'']^{m''}_\alpha\subseteq \left|K_\alpha\varphi\right|^{m''}$ for every $m''$ such that $m\sim_\alpha m''$.  Take $L\in \mathbf{Choice}^{m}_\alpha$ such that there exists $m'\in M$ such that $m\sim_\alpha m'$ and $[L]^{m'}\not\subseteq|K_\alpha\varphi|^{m'}$. This implies that $[L]^{m'''}\not\subseteq\left|\phi\right|^{m'''}$ for some $m'''$ such that $m'\sim_\alpha m'''$. Now, transitivity of $\sim_\alpha$ implies that $m\sim_\alpha m'''$. Therefore, the assumption that $\mathcal{M},\left\langle m,h\right\rangle\models \odot^{\mathcal{S}}_\alpha  \varphi$ implies that there is $L'\in \mathbf{Choice}_\alpha^{m}$ such that $L\prec_s L'$ and, if  $L''=L'$ or $L'\preceq_s L''$, then  $[L'']^{m''}_\alpha\subseteq \left|\varphi\right|^{m''}$ for every $m''$ such that $m\sim_\alpha m''$. By definition of epistemic clusters and transitivity of $\sim_\alpha$, this last clause implies that if  $L''=L'$ or $L'\preceq_s L''$ then  $[L'']^{m''}_\alpha\subseteq \left|K_\alpha\varphi\right|^{m''}$ for every $m''$ such that $m\sim_\alpha m''$. Thus, $L'$ attests to the fact that $\mathcal{M},\left\langle m,h\right\rangle\models \odot^{\mathcal{S}}_\alpha  \left(K_\alpha\varphi\right)$.

\item To see that $\mathcal{M}\models (SuN)$, take $\left\langle m,h\right\rangle$ such that $\mathcal{M},\left\langle m,h\right\rangle\models K_\alpha \square \varphi$.  Take $L\in \mathbf{Choice}^{m}_\alpha$, and let $m'\in M$ be such that $m\sim_\alpha m'$ (which means that there exist $j\in H_{m}$, $j'\in H_{m'}$ such that $\left\langle m,j\right\rangle \sim_\alpha \left\langle m',j'\right\rangle$). Condition ($\mathtt{Unif-H}$) ensures that there exists $h'\in H_{m'}$ such that $\left\langle m,h\right\rangle \sim_\alpha \left\langle m',h'\right\rangle$. The assumption that  $\mathcal{M},\left\langle m,h\right\rangle\models K_\alpha \square \varphi$  then implies that $\mathcal{M},\left\langle m',h'\right\rangle\models \square \varphi$. Thus, for any $h''\in [L]^{m'}_\alpha$, the fact that $h''\in H_{m'}$ yields that $\mathcal{M},\left\langle m',h''\right\rangle\models \varphi$. Therefore, for all $L\in \mathbf{Choice}^{m}_\alpha$ and $m'$ such that $m\sim_\alpha m'$, $[L]^{m'}_\alpha\subseteq \left|\phi\right|^{m'}$, which vacuously implies that $\mathcal{M},\left\langle m,h\right\rangle\models \odot^{\mathcal{S}}_\alpha  \varphi$.

\item To see that $\mathcal{M}\models (s.Oic)$, take $\left\langle m,h\right\rangle$ such that $\mathcal{M},\left\langle m,h\right\rangle\models \odot^{\mathcal{S}}_\alpha  \varphi$. This implies that there exists $L\subseteq H_{m}$ such that $[L]^{m''}_\alpha \subseteq \left|\phi\right|^{m''}$ for every $m''\in M$ such that $m\sim_\alpha m''$. Since $\sim_\alpha$ is reflexive, $[L]^{m}_\alpha \subseteq \left|\phi\right|^{m}$. Now, take $h_{0}\in L$. Let $\left\langle m', h' \right\rangle $ be an index such that $\left\langle m, h_{0}\right\rangle \sim_\alpha \left\langle m',h'\right\rangle$. From the definition of epistemic clusters, $h'\in [L]^{m'}_\alpha$, so the fact that $[L]^{m'}_\alpha \subseteq \left|\phi\right|^{m'}$ implies that $\mathcal{M},\left\langle m',h'\right\rangle \models \varphi$.  Therefore, history $h_{0}\in H_{m}$ is such that, for every  $\left\langle m', h'\right\rangle$ with $\left\langle m, h_{0}\right\rangle \sim_\alpha \left\langle m',h'\right\rangle$, $\mathcal{M},\left\langle m',h'\right\rangle\models \varphi$. This means that $\mathcal{M},\left\langle m,h_{0}\right\rangle \models K_\alpha \varphi$, which implies that $\mathcal{M},\left\langle m,h\right\rangle \models \Diamond K_\alpha \varphi$.

\item To see that $\mathcal{M}\models(s.Cl)$, take $\left\langle m_*,h_*\right\rangle$  such that $\mathcal{M},\left\langle m_*,h_*\right\rangle\models \odot^{\mathcal{S}}_\alpha  \varphi$. Let $\left\langle m,j\right\rangle$ be such that $\left\langle m_*,h_*\right\rangle\sim_\alpha \left\langle m,j\right\rangle$. Take $h\in H_m$. We want to show that, for every $L\in \mathbf{Choice}^{m}_\alpha$ such that $[L]^{m'}\not\subseteq\left|\phi\right|^{m'}$ (for some $m'$ such that $m\sim_\alpha m'$), there is $L'\in \mathbf{Choice}_\alpha^{m}$ such that $L\prec_s L'$ and, if  $L''=L'$ or $L'\preceq_s L''$, then  $[L'']^{m''}_\alpha\subseteq \left|\phi\right|^{m''}$ for every $m''$ such that $m\sim_\alpha m''$.  Take $L\in \mathbf{Choice}^{m}_\alpha$ such that there exists $m'\in M$ such that $m\sim_\alpha m'$ and $[L]^{m'}\not\subseteq\left|\phi\right|^{m'}$. Let $N_L$ be an action in $\mathbf{Choice}^{m_*}_\alpha$ such that $N_L\subseteq [L]^{m_*}_\alpha$, where we know that such an action exists in virtue of $(\mathtt{Unif-H})$ and $(\mathtt{OAC})$. Notice that transitivity of $\sim_\alpha$ entails that $\left[N_L\right]^{o}_\alpha = [L]^{o}_\alpha$ for any moment $o$, so that $\left[N_L\right]^{m'}_\alpha\not\subseteq\left|\phi\right|^{m'}$. Since  $\mathcal{M},\left\langle m_*,h_*\right\rangle\models \odot^{\mathcal{S}}_\alpha  \varphi$,  there must exist $N\in \mathbf{Choice}^{m_*}_\alpha$ such that $N_L\prec_s N$ and, if  $N'=N $ or $N\preceq_s N'$, then $[N']^{m''}_\alpha\subseteq \left|\phi\right|^{m''}$ for every  $m''$ such that $m_*\sim_\alpha m''$. Now, let $L_{N}$ be an action in $\mathbf{Choice}^{m}_\alpha$ such that $L_N\subseteq [N]^{m}_\alpha$ (which implies that $\left[L_N\right]^{o}_\alpha = [N]^{o}_\alpha$ for any moment $o$). We claim that $L\prec_s L_N$, and show our claim with the following argument: let $m''\in M$ be such that $m\sim_\alpha m''$, and take $S\in \mathbf{State}^{m''}_\alpha$; on the one hand, ($\star$) $[L]^{m''}_\alpha\cap S =\left[N_L\right]^{m''}_\alpha\cap S\leq [N]^{m''}_\alpha\cap S=\left[L_N\right]^{m''}_\alpha\cap S$; on the other hand, we know that there exist a moment $m'''$ and a state $S_0\in \mathbf{State}^{m'''}_\alpha$ such that $m_*\sim_\alpha m'''$ and such that $[N]^{m'''}_\alpha\cap S_0\not\leq\left[N_L\right]^{m'''}_\alpha\cap S_0$; therefore, $(\star\star)$  $\left[L_N\right]^{m'''}_\alpha\cap S_0=[N]^{m'''}_\alpha\cap S_0\not\leq\left[N_L\right]^{m'''}_\alpha\cap S_0=[L]^{m'''}_\alpha\cap S_0$. Together, $(\star)$ and $(\star\star)$ entail that $L\prec_s L_N$, proving our claim. Now, let $L''\in \mathbf{Choice}^m_\alpha$ be such that $L''=L_N$ or $L_N\preceq_sL''$. If $L'' =L_N$, then $[L'']^{m''}_\alpha=[N]^{m''}_\alpha\subseteq \left|\phi\right|^{m''}$ for every $m''$ such that $m\sim_\alpha m''$. If $L_N\prec_sL''$, then an argument similar to the one used to show that our claim was true renders that there is an action $N_{L''}\in \mathbf{Choice}^{m_*}_\alpha$ such that $N_{L''}\subseteq [L'']^{m_*}_\alpha$ and $N\preceq_s N_{L''}$. Thus, $[L'']^{m''}_\alpha=[N_{L''}]^{m''}_\alpha\subseteq \left|\phi\right|^{m''}$. With this, we have shown that $\mathcal{M},\left\langle m,h\right\rangle\models \odot^{\mathcal{S}}_\alpha \varphi$ for every $h\in H_m$, so that $\mathcal{M},\left\langle m,j\right\rangle\models \square\odot^{\mathcal{S}}_\alpha  \varphi$. But $\left\langle m,j\right\rangle$ was an arbitrary index such that $\left\langle m_*,h_*\right\rangle\sim_\alpha \left\langle m,j\right\rangle$. Thus, $\mathcal{M},\left\langle m_*,h_*\right\rangle\models K_\alpha\square\odot^{\mathcal{S}}_\alpha  \varphi$.

 \item Let us show that $\mathcal{M}\models(ConSO)$. First of all, let us show that, for all $L, L'\in\mathbf{Choice}_\alpha^m$, if $L\preceq_s L'$, then $L\preceq L'$. Take $L, L'\in\mathbf{Choice}_\alpha^m$. If $L\preceq_s L'$, then, for each $m'$ such that $m\sim_\alpha m'$, $\mathbf{Value}(h)\leq \mathbf{Value}(h')$ for every $h\in[L]^{m'}_\alpha, h'\in [L']^{m'}_\alpha$. Reflexivity of $\sim_\alpha$ implies both that $m\sim_\alpha m'$ and that  $L\subseteq [L]^{m}_\alpha$ and $L'\subseteq [L']^{m}_\alpha$. Therefore, for all $h''\in L$ and $h'''\in L'$,  $\mathbf{Value}(h'')\leq \mathbf{Value}(h''')$, which implies that $L\preceq L'$.   

Now, let $\langle m, h\rangle$ be an index. Assume for a contradiction that ($\star$) $\mathcal{M}, \left\langle m, h\right\rangle \models \odot_\alpha^ {\mathcal{S}} \varphi$ and that ($\star\star$) $\mathcal{M}, \left\langle m, h\right\rangle \models \odot_\alpha  \lnot \varphi$. On the one hand, assumption $(\star)$ implies that there is $L_*\in\mathbf{Choice}_\alpha^m$ such that $L_*\subseteq \left|\phi\right|^m$. Thus, assumption ($\star\star$) yields that there is $L_*'\in \mathbf{Choice}_\alpha^m$ such that $L_*\prec L_*'$ and, if $N=L_*'$ or $L_*'\preceq N$, then $N\subseteq |\lnot\phi|^m$. In particular,  $L_*'\subseteq |\lnot\phi|^m$. Assumption ($\star$) then implies that there is  $L_*''\in\mathbf{Choice}_\alpha^m$ such that $L_*'\prec_s L_*''$ and, if $N=L_*''$ or $L_*''\preceq_s N$, then $N\subseteq [N]^m_\alpha\subseteq \left|\phi\right|^m$. In particular, $L_*''\subseteq \left|\phi\right|^m$. On the other hand, by the first observation in the proof, the fact that $L_*'\prec_s L_*''$ implies that $L_*'\preceq L_*''$, so that assumption ($\star\star$) yields that $L_*''\subseteq |\lnot\phi|^m$, which contradicts the previously shown fact that $L_*''\subseteq \left|\phi\right|^m$. Thus, $\mathcal{M}, \left\langle m, h\right\rangle \models \odot_\alpha^ {\mathcal{S}} \varphi \to  \lnot\odot_\alpha  \lnot \varphi$ for every index $\left\langle m, h\right\rangle$, so that $\odot_\alpha^ {\mathcal{S}} \varphi \to  \lnot\odot_\alpha  \lnot \varphi$ is indeed valid. 

\item It is clear that the rules of inference \emph{Modus Ponens}, Substitution, and Necessitation for the modal operators all preserve validity. 
\end{proof}

\subsection{Completeness}

As mentioned in the main body of the paper, whether $\Lambda_R$ is complete with respect to the class of \emph{kiobt}-models is still an open problem. However, the proof system $\Lambda_{R}'$---obtained from $\Lambda_{R}$ by eliminating $(ConSO)$ in Definition~\ref{axiomres1}---is sound and complete with respect to the class of bi-valued \emph{kiobt}-models (Definition~\ref{multikb}). Soundness follows from Proposition~\ref{soundres}, and the proof of completeness is obtained by integrating the proofs of completeness in \cite{abarca2019logic} and \cite{abarca2022int}. More precisely, the proof of completeness will be sketched below as a two-step process. First, I introduce a Kripke semantics for \emph{bi-valued} \emph{IEAUST}, where the formulas of $\mathcal{L}_{\textsf{R}}$ are evaluated on bi-valued Kripke-\emph{kios}-models (Definition~\ref{modelsKripkeres}). Completeness of ${\Lambda_R}$' with respect to the class of these structures is shown  via the well-known technique of canonical models. Secondly, a truth-preserving correspondence between bi-valued Kripke-\emph{kios}-models and a sub-class of bi-valued \emph{kiobt}-models is used to prove completeness with respect to bi-valued \emph{kiobt}-models via completeness with respect to bi-valued Kripke-\emph{kios}-models.

A Kripke semantics for \emph{IEAUST} is defined as follows:

\begin{definition}[Bi-valued Kripke-\emph{kios}-frames  \& models]
\label{modelsKripkeres}
A tuple 
\[\left\langle W, Ags, R_\square,\mathbf{\mathbf{Choice}}, \left\{\approx_\alpha\right\}_{\alpha\in Ags}, \left\{R_\alpha ^I\right\}_{\alpha\in Ags},  \mathtt{Value}_{\mathcal{O}}, \mathtt{Value}_{\mathcal{S}} \right\rangle\] is called a \emph{bi-valued} Kripke-\emph{kios}-frame iff
\begin{itemize}

\item $W$ is a set of possible worlds. $R_\square$ is an equivalence relation over $W$. For $w\in W$, the class of $w$ under $R_\square$ is denoted by $\overline{w}$. $\mathtt{Choice}$ is a function that assigns to each $\alpha\in Ags$ and $\square$-class $\overline{w}$ a partition $\mathtt{Choice}_\alpha^{\overline{w}}$ of $\overline{w}$ given by an equivalence relation denoted by $R_\alpha^{\overline{w}}$. $\mathtt{Choice}$ must satisfy the following constraint:
\begin{itemize}
\item $\mathtt{(IA)_K}$ For all $w\in W$, each function $s:Ags\to 2^{\overline{w}}$ that maps $\alpha$ to a member of $\mathtt{Choice}^{\overline{w}}_\alpha$ is such that $\bigcap_{\alpha \in Ags} s(\alpha) \neq \emptyset$ (where the set of all functions $s$ that map $\alpha$ to a member of $\mathtt{Choice}^{\overline{w}}_\alpha$ is denoted by $\mathtt{Select}^{\overline{w}}$). 

\end{itemize}
For $\alpha\in Ags$, $w\in W$, and $v\in \overline{w}$, the class of $v$ in the partition $\mathtt{Choice}^{\overline{w}}_\alpha$ is denoted by $\mathtt{Choice}^{\overline{w}}_\alpha(v)$. Now, for $\beta\in Ags$ and $w\in W$, $\mathtt{State}_\beta^{\overline{w}}:=\left\{S\subseteq \overline{w} ; S=\bigcap_{\alpha \in Ags-\left\{\beta\right\}} s(\alpha), \mbox{ for } s\in \mathtt{Select}^{\overline{w}}\right\}$, where $\mathtt{Select}^{\overline{w}}$ denotes the set of all selection functions at $\overline{w}$ (i.e., functions that assign to each $\alpha$ a member of $\mathtt{Choice}^{\overline{w}}_\alpha$).


\sloppy

\item For all $\alpha \in Ags$, $\approx_\alpha$ is an (epistemic) equivalence relation on $W$. The following conditions must be satisfied: 
\begin{itemize}
\item $\mathtt{(OAC)_{\mathtt{K}}}$ For all $\alpha\in Ags$, $w\in W$, and $v\in \overline{w}$, $v\approx_\alpha u$ for every $u\in \mathtt{Choice}^{\overline{w}}_\alpha(v)$.
\item $\mathtt{(Unif-H)_{\mathtt{K}}}$ For all $\alpha\in Ags$, if $v, u\in W$ are such that  $v\approx_\alpha u$, then for all $v'\in \overline{v}$ there exists $u'\in \overline{u}$ such that $v'\approx_\alpha u'$. 
\end{itemize}

For $\alpha\in Ags$ and $w\in W$,  $\alpha$'s \emph{\emph{ex ante} information set at $w$} is defined as $\pi_\alpha^\square[w]:=\left\{v; w\approx_\alpha \circ R_\square v \right\}$, which by frame condition $(\mathtt{Unif-H})_\mathtt{K}$ coincides with the set $\left\{v; w R_\square\circ \approx_\alpha  v \right\}$. To clarify, $(\mathtt{Unif-H})_\mathtt{K}$ implies that $R_\square\circ \approx_\alpha= \approx_\alpha \circ R_\square$. Thus, $\approx_\alpha \circ R_\square$ is an equivalence relation such that $\pi_\alpha^\square[w]=\pi_\alpha^\square[v]$ for every $w,v \in W$ such that $w\approx_\alpha \circ R_\square v$.

\item For $\alpha\in Ags$, $R_\alpha^I$ is a serial, transitive, and euclidean relation on $W$ such that $R_\alpha^I\subseteq \approx_\alpha \circ R_\square$ and such that the following condition is satisfied: 
\begin{itemize}
    \item $\mathtt{(Den)_K}$ For all $v, u\in W$ such that $v \approx_\alpha \circ R_\square  u$, there exists $z\in W$ such that $vR_\alpha^I z$ and $uR_\alpha^I z$.
\end{itemize}

For $\alpha\in Ags$, $R_\alpha^{I+}$ denotes the reflexive closure of $R_\alpha^I$. For $w\in W$, $w\uparrow_{R_\alpha^{I+}}$ denotes the set $\left\{v\in W; w R_\alpha^{I+} v\right\}$.

For $w, v\in W$, I write $\overline{w} \approx_\alpha \overline{v}$ iff there exist $w'\in\overline{w}$ and $v'\in\overline{v}$ such that $w'\approx_\alpha v'$. For $w, v\in W$ such that $\overline{w} \approx_\alpha \overline{v}$ and $L\in\mathtt{Choice}^{\overline{w}}_\alpha$, $L$'s epistemic cluster at $\overline{v}$ is the set $\sembracki{L}_\alpha^{\overline{v}}:=\left\{u\in {\overline{v}} ; \mbox{ there is } o\in L \ \textnormal{ such that }\ o \approx_\alpha u \right\}.$

\fussy

\item $\mathtt{Value}_{\mathcal{O}}$ and $\mathtt{Value}_{\mathcal{S}}$ are functions that independently assign to each world $w\in W$ a real number. 

These functions are used to define an objective ordering $\preceq$ and a subjective ordering $\preceq_s$ of choices. Formally, for $\alpha\in Ags$ and $w\in W$, one first defines
two general orderings $\leq$ and $\leq_s$ on $2^{W}$ by the rules:
$X\leq Y$ iff $\mathtt{Value}_{\mathcal{O}}(w) \leq \mathtt{Value}_{\mathcal{O}}(w')$ for all $w\in X$ and  $w'\in Y$; and $X\leq_s Y$ iff $\mathtt{Value}_{\mathcal{S}}(w) \leq \mathtt{Value}_{\mathcal{S}}(w')$ for all $w\in X$ and  $w'\in Y$. An objective dominance ordering $\preceq$ is then defined on $\mathtt{Choice}_\alpha^{\overline{w}}$ by the rule: $L\preceq L'$ iff  $L\cap S \leq L'\cap S$ for every $S\in \mathtt{State}_\alpha^{\overline{w}}$.  In turn, a subjective dominance ordering $\preceq_s$ is then defined on $\mathtt{Choice}_\alpha^{\overline{w}}$ by the rule: $L\preceq_s L'$ iff  $\sembracki{L}^{\overline{v}}_\alpha\cap S \leq_s \sembracki{L'}^{\overline{v}}_\alpha\cap S$ for every $v$ such that $w\approx_\alpha v$ and every $S\in \mathtt{State}_\alpha^{\overline{v}}$. I write  $L\prec L'$ iff $L\preceq L'$ and $L'\npreceq L$, and I write $L\prec_s L'$ iff $L\preceq_s L'$ and $L'\npreceq_s L$, so that
$\mathtt{Optimal}_\alpha^{\overline{w}}:=\left\{L \in \mathtt{Choice}^{\overline{w}}_\alpha ; \textnormal{ there is no } L' \in \mathtt{Choice}^{\overline{w}}_\alpha \textnormal{ s. t. }  L\prec L'\right\}$ 
and $\mathtt{SOptimal}_\alpha^{\overline{w}}:=\left\{L \in \mathtt{Choice}^{\overline{w}}_\alpha ; \textnormal{ there is no } L' \in \mathtt{Choice}^{\overline{w}}_\alpha \textnormal{ s. t. }  L\prec_s L'\right\}$. 
\end{itemize} 
A Kripke-\emph{kios}-model $\mathcal {M}$ consists of the tuple that results from adding a valuation function $\mathcal{V}$ to a Kripke-\emph{kios}-frame, where $\mathcal{V}: P\to 
2^{W}$ assigns to each atomic proposition a set of worlds (recall that $P$ is the set of propositions in $\mathcal{L}_{\textsf{R}}$).
\end{definition}

 Kripke-\emph{kios}-models allow us to evaluate the formulas of $\mathcal L_{\textsf{R}}$ with semantics that are analogous to the ones provided for \emph{kiobt}-models:
 
\begin{definition}[Evaluation rules on Kripke models]
\label{evareskripke}
    Let $\mathcal{M}$ be a Kripke-\emph{kios}-model, the semantics on $\mathcal {M}$ for the formulas of $\mathcal {L}_{\textsf{KO}}$ are defined recursively by the following truth conditions, evaluated at  world $w$: 
\[ \begin{array}{lll}
\mathcal{M}, w \models p & \mbox{iff} &  w \in \mathcal{V}(p) \\

\mathcal{M}, w \models \neg \phi & \mbox{iff} & \mathcal{M}, w \not\models \phi \\

\mathcal{M}, w \models \phi \wedge \psi & \mbox{iff} & \mathcal{M}, w \models \phi \mbox{ and } \mathcal{M}, w \models \psi \\

\mathcal{M}, w \models \Box \phi &
\mbox{iff} & \mbox{for each } v\in \overline{w},\mathcal{M}, v \models \phi \\

\mathcal{M}, w \models [\alpha]
\phi & \mbox{iff} & \mbox{for each } v\in \mathtt{Choice}^{\overline{w}}_\alpha(w), \mathcal{M}, v \models \phi\\

\mathcal{M}, w \models K_{\alpha} \phi &
\mbox{iff} & \mbox{for each } v \mbox{ s. t. }   w \approx_{\alpha}v, \mathcal{M}, v \models \phi\\

\mathcal{M}, w \models I_\alpha\phi &
\mbox{ iff } & \mbox{there exists } x \in \pi_\alpha^\square[w]\mbox{ s. t. }  x\uparrow_{R_\alpha^{I+}} \subseteq |\phi|\\

\mathcal{M},w \models \odot_\alpha  \varphi & \mbox{iff} &\mbox{for all } L\in \mathtt{Choice}^{\overline{w}}_\alpha \mbox{ s. t. } \mathcal{M},v\not\models\varphi \mbox{ for some } v\in L, \mbox{ there is}\\ &&  L'\in \mathtt{Choice}_\alpha^{\overline{w}} \mbox{ s. t. } L\prec L'  \mbox{ and, if } L''=L' \mbox{ or } L'\preceq_s L'', \\ &&  \mbox{then } \mathcal{M},w'\models \varphi \mbox{ for every }  w'  \in L''_\alpha\\

\mathcal{M},w \models \odot^{\mathcal{S}}_\alpha  \varphi & \mbox{iff} &\mbox{for all } L\in \mathtt{Choice}^{\overline{w}}_\alpha \mbox{ s. t. } \mathcal{M},v\not\models\varphi \mbox{ for some } w' \mbox{ s. t. } w\approx_\alpha w' \\ && \mbox{and some } v\in \sembracki{L}^{w'}_\alpha, \mbox{ there is } L'\in \mathtt{Choice}_\alpha^{\overline{w}} \mbox{ s. t. } L\prec_s L' \\ &&  \mbox{and, if }  L''=L' \mbox{ or } L'\preceq_s L'', \mbox{then } \mathcal{M},w'''\models \varphi \mbox{ for every }  w''  \\&&   \mbox{s. t. } \overline{w}\approx_\alpha \overline{w''} \mbox{ and every } w'''\in \sembracki{L''}^{w''}_\alpha,
\end{array} \]
where I write $|\phi|$ to refer to the set $\left\{w\in W;\mathcal{M},w \models \phi\right\}$. Satisfiability, validity on a frame, and general validity are defined as usual. 
\end{definition}

A truth-preserving correspondence between Kripke-\emph{kios}-models and \emph{kiobt}-models is shown as follows: 

\begin{definition}[Associated \emph{kiobt}-frame] \label{treeres}

Let \[\mathcal{F}=\left\langle W, Ags, R_\square, \mathtt{Choice},  \left\{\mathtt{\approx}_\alpha\right\}_{\alpha\in Ags}, \left\{R_\alpha^I\right\}_{\alpha\in Ags}
, \mathtt{Value}_{\mathcal{O}}, \mathtt{Value}_{\mathcal{S}}\right\rangle\] be a bi-valued Kripke-\emph{kios}-frame.

Then $ \mathcal{F}^T:=\left\langle M_W,  \sqsubset, Ags, \mathbf{Choice},\left\{\sim_\alpha\right\}_{\alpha\in Ags}, \tau, \mathbf{Value}_{\mathcal{O}}, \mathbf{Value}_{\mathcal{S}} \right\rangle$ is called the bi-valued \emph{kiobt}-frame associated with $\mathcal{F}$ iff
\begin{itemize}
\item $M_W, \sqsubset, \mathbf{Choice}$,  $\left\{\sim_\alpha\right\}_{\alpha\in Ags}$, and $\tau$ are defined just as in Definition~11 in \cite{abarca2022int}. 

\item  $\mathbf{Value}_{\mathcal{O}}$ and $\mathbf{Value}_{\mathcal{S}}$ are defined by the following rules: for $h_v\in H$, $\mathbf{Value}_{\mathcal{O}}(h_v)=\mathtt{Value}_{\mathcal{O}}(v)$, and  $\mathbf{Value}_{\mathcal{S}}(h_v)=\mathtt{Value}_{\mathcal{S}}(v)$. 
\end{itemize}
\end{definition} 

\begin{proposition}\label{avion2res}
Let $\mathcal{F}$ be a bi-valued Kripke-\emph{kios}-frame. Then $\mathcal{F}^T$ is a bi-valued \emph{kiobt}-frame, indeed. 
\end{proposition}

\begin{proof}
Follows from Proposition~2 in \cite{abarca2022int} and Definition~\ref{treeres}.
\end{proof}

\begin{lemma}\label{kriptotreeres}
Let $\mathcal{M} $ be a bi-valued Kripke-\emph{kios}-model, and let $\mathcal{M}^T$ be its associated  bi-valued  \emph{kiobt}-model.  For all $\alpha\in Ags$, $w\in W$, and $L, N\in \mathtt{Choice}^{\overline{w}}_\alpha$, the following conditions hold:
\begin{enumerate}[(a)]
\item \label{itm:asu} $L\preceq N$ iff $L^T\preceq N^T$ and $L\prec N$ iff  $L^T\prec N^T$.
\item \label{itm:bsu} $L\preceq_s N$ iff $L^T\preceq_s N^T$ and $L\prec_s N$ iff   $L^T\prec_s N^T$.
\item \label{itm:csu} $L\in \mathtt{Optimal}^{\overline{w}}_\alpha$ iff $L^T\in \mathbf{Optimal}^{\overline{w}}_\alpha$.

\item \label{itm:dsu} $L\in \mathtt{S-Optimal}^{\overline{w}}_\alpha$ iff $L^T\in \mathbf{S-Optimal}^{\overline{w}}_\alpha$.
\end{enumerate}
\end{lemma}
\begin{proof}
The reader is referred to the proof of Proposition~4 in \url{https://doi.org/10.48550/arXiv.1903.10577} for a proof.
\end{proof}

\begin{proposition}[Truth-preserving correspondence]\label{avion1res}
Let $\mathcal{M}$ be a bi-valued Kripke-\emph{kios}-model, and let $\mathcal{M}^T$ be its associated bi-valued \emph{kiobt}-model. For all $\phi$ of $\mathcal{L}_{\textsf{R}}$ and $w\in W$,  $\mathcal{M} ,w\models\phi$ iff $\mathcal{M} ^T,\left\langle\overline{w},h_w\right\rangle\models\phi$.  
\end{proposition}

\begin{proof}

We proceed by induction on the complexity of $\phi$. For the base case, the cases of Boolean connectives, and the cases of all modal operators except $I_\alpha$, the proofs are exactly the same as their analogs' in Proposition~4 in \url{https://doi.org/10.48550/arXiv.1903.10577}. For the case of $I_\alpha$, the proof is the same as its analog in  Proposition~3 in \cite{abarca2022int}.  

\end{proof}

Thus, completeness with respect to bi-valued \emph{kiobt}-models is proved with Propositions~\ref{completeres} and \ref{completenessres} below.  

\begin{proposition}[Completeness w.r.t. bi-valued Kripke-\emph{kios}-models] \label{completeres} The proof system ${\Lambda_R}$' is complete with respect to the class of bi-valued Kripke-\emph{kios}-models.
\end{proposition}

\sloppy
\begin{proof}
Completeness with respect to bi-valued Kripke-\emph{kios}-models is shown via canonical models. To be precise, one defines a structure $\mathcal{M}:=\left\langle W^{\Lambda_R'}, R_\square, \mathtt{Choice}, \left\{\mathtt{\approx}_\alpha\right\}_{\alpha\in Ags}, \left\{R_\alpha^I\right\}_{\alpha\in Ags} \mathtt{Value}_{\mathcal{O}}, \mathtt{Value}_{\mathcal{S}},\mathcal{V} 
\right\rangle$, where $W^{\Lambda_R'}=\left\{w ;w  \mbox{ is a } {\Lambda_R'}\textnormal{-MCS}\right\}$, where $R_\square, \mathtt{Choice}, \left\{\mathtt{\approx}_\alpha\right\}_{\alpha\in Ags}$, $\left\{R_\alpha^I\right\}_{\alpha\in Ags}$, and $\mathcal{V}$  are defined just as in Definition~12 in \cite{abarca2022int}, and where $\mathtt{Value}_\mathcal{O}$ and $\mathtt{Value}_\mathcal{S}$ are defined as follows: for $\alpha\in Ags$ and $w\in W^\Lambda$, one first defines $\Sigma_\alpha^{w}:=\{[\alpha ]\phi ;\odot[\alpha]\phi \in w\}$ and  
   $\Gamma_\alpha^{w}:=\{K_\alpha\phi; \odot_{\mathcal{S}}[\alpha]\phi \in w\}$. Then, taking $\Sigma^{w}=\bigcup_{\alpha\in Ags} \Sigma_\alpha^{w}$ and $\Gamma^{w}= \bigcup_{\alpha\in Ags}$, the deontic functions are given by

\[\begin{array}{ll}
    \mathtt{Value} _\mathcal{O}(w) &=  \left\{\begin{array}{lll} 1 \mbox{ iff }  \Sigma^{w}  \subseteq w,& \\ 0 \mbox{ otherwise}.&\end{array}\right.\\
    \mathtt{Value} _\mathcal{S}(w) &=  \left\{\begin{array}{lll} 1 \mbox{ iff }  \Gamma^{w} \subseteq w,& \\ 0 \mbox{ otherwise}.&\end{array}\right.\\
\end{array}\]

The canonical structure $\mathcal{M}$ is shown to be a bi-valued Kripke-\emph{kios}-model just as in Proposition~4 in \cite{abarca2022int}. Then, the so-called \emph{truth lemma} is shown by merging Lemma~2 in \cite{abarca2022int} and Lemma~4 in \url{https://doi.org/10.48550/arXiv.1903.10577}. This renders completeness with respect to bi-valued Kripke-\emph{kios}-models.
\end{proof}

\begin{proposition} [Completeness w.r.t.  bi-valued \emph{kiobt}-models] \label{completenessres} The proof system ${\Lambda_R}$' is complete with respect to the class of bi-valued \emph{kiobt}-models.
\end{proposition}

\begin{proof}
Let $\phi$ be a ${\Lambda_R'}$-consistent formula of $\mathcal{L}_{\textsf{R}}$. Proposition~\ref{completeres} implies that there exists a bi-valued Kripke-\emph{kios}-model $\mathcal{M}$ and a world $w$ in its domain such that $\mathcal{M} , w \models \phi$. Proposition~\ref{avion1res} then ensures that the  bi-valued \emph{kiobt}-model $\mathcal{M}^T$ associated with $\mathcal{M}$ is such that $\mathcal{M}^T, \left\langle\overline{w},h_w  \right\rangle\models \phi$. 
\end{proof}

Therefore, Proposition~\ref{soundres} and Proposition~\ref{completenessres} imply that the following result, appearing in the main body of the paper, is true:

\soundcompres*

\end{subappendices}
\end{document}